\documentclass[12pt,reqno]{amsart}
\usepackage{fullpage}

\newtheorem{theorem}{Theorem}[section]

\newtheorem{cor}[theorem]{Corollary}
\newtheorem{conj}[theorem]{Conjecture}

\newtheorem{prob}[theorem]{Problem}
\usepackage{graphicx}
\usepackage{color}
\usepackage[dvipsnames]{xcolor}
\usepackage{subfigure}
\usepackage{amssymb}
\usepackage{amsmath,mathrsfs}
\usepackage{colonequals}
\usepackage{hyperref}

\theoremstyle{definition}
\newtheorem{definition}[theorem]{Definition}

\newtheorem{example}[theorem]{Example}

\newtheorem{remark}[theorem]{Remark}

\renewcommand{\subset}{\subseteq}

\renewcommand{\epsilon}{\varepsilon}

\newcommand{\abs}[1]{\left|#1\right|}                   
\newcommand{\vnorm}[1]{\left\|#1\right\|}    
\newcommand{\vnormt}[1]{\left\|#1\right\|}    

\newcommand{\Z}{\mathbb{Z}}                             

\newcommand{\E}{\mathbb{E}}

\renewcommand{\d}{\mathrm{d}}

\renewcommand{\P}{\mathbb{P}}
\newcommand{\R}{\mathbb{R}}

\newcommand{\embolden}[1]{\textbf {#1}}

\newcommand{\sdimn}{n}

\begin{document}

\title{Noise Stability of Ranked Choice Voting}

\author{Steven Heilman}
\address{Department of Mathematics, University of Southern California, Los Angeles, CA 90089-2532}
\email{stevenmheilman@gmail.com}
\date{\today}
\thanks{S. H. is Supported by NSF Grant CCF 1911216}

\begin{abstract}
We conjecture that Borda count is the ranked choice voting method that best preserves the outcome of an election with randomly corrupted votes, among all fair voting methods with small influences satisfying the Condorcet Loser Criterion.  This conjecture is an adaptation of the Plurality is Stablest Conjecture to the setting of ranked choice voting.  Since the plurality function does not satisfy the Condorcet Loser Criterion, our new conjecture is not directly related to the Plurality is Stablest Conjecture.  Nevertheless, we show that the Plurality is Stablest Conjecture implies our new Borda count is Stablest conjecture.  We therefore deduce that Borda count is stablest for elections with three candidates when the corrupted votes are nearly uncorrelated with the original votes.  We also adapt a dimension reduction argument to this setting, showing that the optimal ranked choice voting method is ``low-dimensional.''

The Condorcet Loser Criterion asserts that a candidate must lose an election if each other candidate is preferred in head-to-head comparisons.  Lastly, we discuss a variant of our conjecture with the Condorcet Winner Criterion as a constraint instead of the Condorcet Loser Criterion.  In this case, we have no guess for the most stable ranked choice voting method.
\end{abstract}


\maketitle
\setcounter{tocdepth}{1}
\tableofcontents
%
%
%
%
%


\section{Introduction}

A basic question in social choice theory is: what is the best voting method?  In this paper, a voting method is a function $f$ whose input is a set of $n$ votes and whose output is the winner of the election.  There are many possible descriptions of a ``best'' voting method.  For example, Rousseau suggested in the 1700s that the best two-candidate voting method maximizes the number of votes that agree with the election's outcome \cite{rousseau62}.  In this sense, the majority function is the best voting method \cite[Theorem 2.33]{odonnell14}, at least for an add number of voters satisfying the impartial culture assumption (that voters' preferences are independent and uniformly random).

Through the years, various other notions of ``best'' voting method have been studied, e.g. by game theorists such as Shapley, Shubik and Banzhaf in the 1950s and 1960s \cite{shapley54,banzhaf65}.  Statements such as Arrow's Impossibility Theorem \cite{arrow50,arrow02,arrow11} imply there is no ``best'' voting method, since no voting method satisfies a short list of reasonable assumptions.  In the last few decades, computational aspects of social choice theory have been investigated (such as the complexity of computing the winner of an election or of finding voting strategies) \cite{brandt16}.  However, this article does not consider that topic.  Contemporaneously, discrete Fourier analysis methods have been incorporated into social choice theory, sometimes motivated by applications in computational complexity \cite{kahn88,kalai02,mossel10,mossel12b}.

One recent highlight in the latter theory is the Majority is Stablest Theorem \cite{mossel10}.

\begin{theorem}[\embolden{Majority is Stablest, Informal Version}, {\cite[Theorem 4.4]{mossel10}}]\label{misinf}
Suppose we run an election with a large number $n$ of voters and two candidates.  In this election, we make the following assumptions:
\begin{itemize}
\item[(i)] Voters cast their votes randomly, independently, with equal probability of voting for either candidate.
\item[(ii)] Each voter has a small influence on the outcome of the election.  (That is, all influences from Definition \ref{infdef} are small.)
\end{itemize}
Then the majority function is the balanced voting method that best preserves the outcome of the election, when votes have been corrupted independently.

\end{theorem}
We say an election method is balanced if each candidate has an equal chance of winning the election.

We say that votes are corrupted independently with probability $0<\epsilon<1$ if, for any given vote, with probability $1-\epsilon$ that vote remains the same, and with probability $\epsilon$ that vote is set to be equally likely to be any of the candidates, independently of all other votes and all other corruptions.

Since the majority function is optimal in the sense of Theorem \ref{misinf}, we say the majority function is the most stable or the most noise stable balanced voting method, among low influence functions.

To see a more general version of Theorem \ref{misinf} with the balanced assumption changed, see \cite[Theorem 4.4]{mossel10}.  If the small influence assumption is removed in Theorem \ref{misinf}, then it is a standard exercise to show that dictator functions are the most stable voting methods.  So, in order to have majority be the most stable voting method, Theorem \ref{misinf} requires an assumption on small influences.  Note also that anti-majority (where the majority preferred candidate loses) also has the same probability of preserving the outcome as the majority function itself, so the term ``Majority is Stablest'' is a bit deceiving.

The original motivation for Theorem \ref{misthm} was proving sharp computational hardness for approximation algorithms for MAX-CUT.  Theorem \ref{misthm} implies that the Goemans-Williamson semidefinite program \cite{goemans95} has an optimal quality of approximation for the MAX-CUT problem among all polynomial time algorithms, assuming the Unique Games Conjecture is true \cite{khot07,khot18}.

The analogue of Theorem \ref{misinf} for $3$ or more candidates is still an open problem, known as the Plurality is Stablest Conjecture, though some cases of the conjecture have been proven \cite{heilman20d}.  In the current setting, each voter votes for exactly one candidate.

\begin{conj}[\embolden{Plurality is Stablest, Informal}, {\cite{khot07,isaksson11}}]\label{pisinf}
Suppose we run an election with a large number $n$ of voters and $k\geq3$ candidates.   In this election, we make the following assumptions:
\begin{itemize}
\item[(i)] Voters cast their votes randomly, independently, with equal probability of voting for each candidate.
\item[(ii)] Each voter has a small influence on the outcome of the election.
\end{itemize}
Then, among all balanced voting methods, the plurality function best preserves the outcome of the election when votes have been corrupted independently.

\end{conj}

As before, we say an election method is balanced if each candidate has an equal chance of winning the election.

In fact, the main result of \cite{heilman20d} shows that the balanced voting method that best preserves the outcome of an election between $k\geq3$ candidates is ``low-dimensional,'' in that it is a function of at most $k-1$ linear functions of the votes.  

\begin{theorem}[\embolden{Plurality is Stablest, Weak Form, Informal}, {\cite{heilman20d}}]\label{pisweak}
Let $k\geq2$.  Under the assumptions of Conjecture \ref{pisinf}, if we denote the votes of the candidates as $x=(x_{1},\ldots,x_{n})\in\{1,\ldots,k\}^{n}$, then the most stable balanced voting method $g\colon\{1,\ldots,k\}^{n}\to\{1,\ldots,k\}$ satisfies: $\forall$ $1\leq i\leq k$, $\exists$ $w^{(i,1)},\ldots,w^{(i,k-1)}\in\R^{n}$ and $\exists$ $h_{i}\colon\R^{k-1}\to\R$ such that
$$1_{\{g(x)=i\}}=h_{i}(\langle x,w^{(i,1)}\rangle,\ldots,\langle x,w^{(i,k-1)}\rangle),\qquad\forall\,x\in\{1,\ldots,k\}^{n}.$$
\end{theorem}
In fact, the conclusion of Theorem \ref{pisweak} also holds with an unbalanced assumption on the voting method.

In particular, when $k=2$, Theorem \ref{pisweak} gives a weakened form of the Majority is Stablest Theorem, stating that the most stable balanced voting method is a function of the form
$$h(\langle x,y^{(1)}\rangle),\qquad\forall\,x\in\{-1,1\}^{n}.$$
The Majority is Stablest Theorem \cite{mossel10} says we can choose $h\colonequals\mathrm{sign}$ and $y^{(1)}\colonequals(1,\ldots,1)$ when $k=2$.  On the other hand, Theorem \ref{pisweak} applies for all $k\geq2$, whereas the original proof of the Majority is Stablest Theorem \cite{mossel10} does not seem to apply except when $k=2$.

Theorem \ref{pisweak} stated here is a discretization of Theorem \ref{pisc} below from \cite{heilman20d}.  Note also that, without the balanced assumption, the Plurality is Stablest Conjecture is false \cite{heilman14}.  However, an optimal unbalanced voting method is still low-dimensional \cite{heilman20d}.  For a more thorough survey of these and related topics, see e.g. \cite{heilman20b}.

As in Theorem \ref{misinf}, the Plurality is Stablest Conjecture asserts that the plurality voting method is the voting method that best preserves an election's outcome among all ``democratic'' voting methods.  This conjecture for $k\geq3$ candidates is closely related to sharp computational hardness results for the MAX-k-CUT problem \cite{isaksson11}.  So, from the perspective of computational complexity, plurality voting is well motivated.

However, from the perspective of social choice theory, the plurality voting method has many undesirable properties.  It would be better to adapt Theorem \ref{misinf} to the setting of ranked choice voting.  Ranked choice voting allows voters a list of preferences rather than a single vote.  Yet, even formulating a conjectural version of Theorem \ref{misinf} for three or more candidates seemed difficult, since certain natural generalizations will just reduce to the Plurality is Stablest Conjecture itself.  (A priori, it is possible that the most stable ranked-choice voting method simply ignores rankings other than the most preferred candidate.)  So, it is natural to impose an additional constraint on the voting methods to ensure they ``actually incorporate'' all rankings of voters.  We find the most natural such constraint is the Condorcet Loser Criterion.  Here a technical difficulty arises since some constraints are not as natural when voters act independently and uniformly at random.  (For example, with a large number $n$ of voters, some imposed constraint might occur with probability tending to zero as $n\to\infty$.  In such a case, that constraint does not really affect the optimal voting methods under consideration, as $n\to\infty$.)

With these issues in mind, our first result is a formulation of a ranked choice analogue of the Plurality is Stablest Conjecture.

\begin{conj}[\embolden{Borda count is Stablest, Informal Version}]\label{bisinf}
Suppose we run a ranked choice election with a large number $n$ of voters and $k\geq3$ candidates.   In this election, we make the following assumptions:
\begin{itemize}
\item[(i)] Voters cast their ranked votes randomly, independently, with equal probability of providing any given ranking of candidates.
\item[(ii)] Each voter has a small influence on the outcome of the election.
\end{itemize}
We also assume that the voting method is balanced (each candidate has an equal chance of winning the election) and satisfies the Condorcet Loser Criterion (CLC), so that a Condorcet loser can never win the election (see Definition \ref{clcdef}).

Then the Borda count voting method is the balanced voting method satisfying CLC that best preserves the outcome of the election, when votes have been corrupted independently.
\end{conj}

We say that votes are corrupted independently with probability $0<\epsilon<1$ if, for any given person's vote, with probability $1-\epsilon$ that vote remains the same, and with probability $\epsilon$ that vote is set to be uniformly random among all possible rankings of candidates, independently of all other votes and all other corruptions.

In the ranked choice setting, there are a few different natural ways to corrupt votes.  In our formulation, we consider a corruption as a resampling among all possible rankings of voters.  Another natural alternative involves swapping the rankings of two uniformly chosen (distinct) candidates.  The invariance principle \cite{mossel10} implies that these notions of vote corruption lead to equivalent versions of Conjecture \ref{bisinf}, up to changing the corruption probability $\epsilon$.

As mentioned above, Conjecture \ref{bisinf} is an adaptation of the Plurality is Stablest Conjecture to the setting of ranked choice voting.  An exercise shows: if we remove the CLC constraint in Conjecture \ref{bisinf}, then the optimal voting method becomes the plurality function itself.  So, some extra constraint is needed in order to get a ``truly'' ranked choice voting analogue of the Plurality is Stablest Conjecture.

Although Conjecture \ref{bisinf} is not directly related to Conjecture \ref{pisinf}, our first result shows that Conjecture \ref{pisinf} implies Conjecture \ref{bisinf}.

\begin{theorem}[\embolden{Plurality is Stablest Implies Borda count is Stablest}]\label{thm0}
Assume that Conjecture \ref{pisinf} holds for some probability of corruption $0<\epsilon<1$.  Then Conjecture \ref{bisinf} holds with a probability of corruption $\epsilon$.
\end{theorem}

By identifying the set of $k!$ rankings of $k$ candidates with the set $\{1,\ldots,k!\}$, we can adapt the main dimension reduction result of \cite{heilman20d} to Conjecture \ref{bisinf}.  That is, we can show the most stable ranked choice voting method between $k\geq3$ candidates in Conjecture \ref{bisinf} is ``low-dimensional,'' in that it is a function of at most $k!+k-1$ linear functions of the votes.

\begin{theorem}[\embolden{Borda count is Stablest, Weak Form, Informal}]\label{bisweak}
Let $k\geq3$.  Under the assumptions of Conjecture \ref{bisinf}, if we denote the votes of the candidates as $x=(x_{1},\ldots,x_{n})\in\{1,\ldots,k!\}^{n}$, then the most stable balanced voting method $g\colon\{1,\ldots,k!\}^{n}\to\{1,\ldots,k\}$ that satisfies CLC must satisfy:  $\forall$ $1\leq i\leq k$, $\exists$ $w^{(i,1)},\ldots,w^{(i,k!+k-1)}\in\R^{n}$ and $\exists$ $h_{i}\colon\R^{k!+k-1}\to\R$ such that
$$1_{\{g(x)=i\}}=h_{i}(\langle x,w^{(i,1)}\rangle,\ldots,\langle x,y^{(i,k!+k-1)}\rangle),\qquad\forall\,x\in\{1,\ldots,k!\}^{n}.$$
\end{theorem}

Theorem \ref{bisweak} is a discretized version of a more formal statement, presented in Theorem \ref{bisdim} below.  That is, Theorem \ref{bisweak} follows by applying the invariance principle \cite{mossel10} to Theorem \ref{bisdim}.

In some sense, the majority function is the simplest voting method for two candidates where each voter has a small influence on the election's outcome.  The Majority if Stablest Theorem reinforces this fact: if a voting method has a complicated description, corresponding e.g. to various inequalities holding or not holding, then there are more opportunities for random vote corruption to change the election's outcome.  Likewise, the plurality function is arguably one of the simplest voting methods for three or more candidates, and the Plurality is Stablest Conjecture reflects this fact.  So, using a posteriori reasoning, it is perhaps unsurprising that the Borda count method could be the most stable ranked choice voting method among voting functions satisfying the Condorcet Loser Criterion.  Complicated voting methods depending on various conditions or inequalities holding lead to more opportunities for random vote corruptions to change an election's outcome.

In summary, it is perhaps not surprising that finding the most stable voting method should correspond to finding voting methods with very simple descriptions.

Before formulating Conjecture \ref{bisinf}, it seemed reasonable that Instant Runoff Voting should be the stablest ranked choice voting method satisfying CLC.  Computer simulations suggest otherwise.  Using again our a posteriori reasoning, the Borda count method has an arguably simpler description than Instant Runoff Voting, so perhaps Borda's superior stability is then less surprising.  After all, Borda count can be viewed as taking the plurality of pairwise comparisons.  Then Conjecture \ref{bisinf} might even be loosely equivalent to the original Plurality is Stablest Conjecture.

$$
\begin{array}{l|ccccc}
\mathrm{Vote\,Corruption\,Probability} & .01 & .02 & .03 & .05 & .1\\
\hline
\mathrm{Majority} & .045 & .064 & .078 & .101 & .144\\
\hline
\hline
\mathrm{Plurality} & .066  & .095 & .118 & .149 & .209\\
\mathrm{Borda} & .066 & .096 & .119 & .150 & .214\\
\mathrm{Kemeny\,Young} & .071 & .098 & .121 & .154 & .215\\
\mathrm{IRV} & .078 & .106 & .127 & .162 & .225\\
\mathrm{Ranked\,Pairs} & .096& .126 & .151 & .185 & .243\\
\mathrm{Copeland} & .117 & .139 & .158 & .188 & .243\\
\end{array}
$$
\begin{center}
Estimated Probabilities that an Election's Outcome is Changed for Various Voting Methods, 3 Candidates (Majority Provided for Reference)
\end{center}

$$
\begin{array}{l|ccccc}
\mathrm{Vote\,Corruption\,Probability} & .01 & .02 & .03 & .05 & .1\\
\hline
\mathrm{Majority} & .045 & .064 & .078 & .101 & .144\\
\hline
\hline
\mathrm{Plurality} & .078  & .115 & .135 & .177& .251\\
\mathrm{Borda} & .084 & .111 & .137 & .177 & .255\\
\mathrm{Kemeny\,Young} & .091 & .126 & .145 & .193& .273\\
\mathrm{IRV} & .103 & .144 & .173 & .220 & .295\\
\mathrm{Copeland} & .159 & .185 & .205 & .240 & .312\\
\end{array}
$$
\begin{center}
Estimated Probabilities that an Election's Outcome is Changed for Various Voting Methods, 4 Candidates (Majority Provided for Reference)
\end{center}

\subsection{Related Work}

This work concerns the probability of an election method preserving its outcome, subject to corruption, as in \cite{khot07,mossel10,isaksson11}.  Some related but different works address the probabilities of violating voting axioms such as \cite{mossel12b,mossel15b} or \cite{xia20,flan22}.

\section{Majority is Stablest}

A function $f\colon\{-1,1\}^{n}\to\{-1,1\}$ is a \textbf{voting method} with $n$ voters and two candidates.  For any $x=(x_{1},\ldots,x_{n})\in\{-1,1\}^{n}$, we think of $x_{i}$ as the vote of person $1\leq i\leq n$ for candidate $x_{i}\in\{-1,1\}$.  Given the votes $x$, the winner of the election is $f(x)$.

\begin{definition}[\embolden{Majority Function}]
Let $n>1$ be an integer.  Define the \textbf{majority function} $\mathrm{Maj}_{n}\colon\{-1,1\}^{n}\to\{-1,1\}$ to be
$$\mathrm{Maj}_{n}(x)\colonequals\mathrm{sign}(x_{1}+\cdots+x_{n}),\qquad\forall\,x=(x_{1},\ldots,x_{n})\in\{-1,1\}^{n}.$$
(We define the sign of $0$ to be $1$, though this definition will not matter in the statements below such as Theorem \ref{misthm}.)
\end{definition}

\begin{definition}[\embolden{Noise Stability}]\label{nsdisdef}
Let $0\leq \rho\leq 1$.  Let $X=(X_{1},\ldots,X_{n})\in\{-1,1\}^{n}$ be a uniformly distributed random variable in $\{-1,1\}^{n}$.  Let $Y=(Y_{1},\ldots,Y_{n})\in\{-1,1\}^{n}$ be a random variable such that $Y_{1},\ldots,Y_{n}$ are independent, and for each $1\leq i\leq n$, $Y_{i}=X_{i}$ with probability $\rho$, and with probability $1-\rho$, $Y_{i}\in\{-1,1\}$ is uniformly random and independent of $X_{1},\ldots,X_{n}$.  We define the \textbf{noise stability} of $f\colon\{-1,1\}\to\R$ with correlation $\rho$ to be
$$S_{\rho}(f)\colonequals\E_{\rho}f(X)f(Y).$$
\end{definition}
Note that $\E_{\rho} X_{i}Y_{i}=\rho$ and $Y_{i}=X_{i}$ with probability $(1+\rho)/2$ for each $1\leq i\leq n$.  Also, in the case that $f\colon\{-1,1\}^{n}\to\{-1,1\}$, we have
$$(1+S_{\rho}f)/2=\P_{\rho}(f(X)=f(Y)).$$

It is a standard exercise to deduce from the Central Limit Theorem that
$$
\lim_{n\to\infty}S_{\rho}(\mathrm{Maj}_{n})=\frac{2}{\pi}\sin^{-1}(\rho),\qquad\forall\,\rho\in[0,1].
$$
%
%

\begin{definition}[\embolden{Influence}]\label{infdef}
Let $f\colon\{-1,1\}^{n}\to\R$.  Fix $1\leq i\leq n$.  Denote $D_{i}f(x)\colonequals[f(x)-f(x_{1},\ldots,x_{i-1},-x_{i},x_{i+1},\ldots,x_{n})]/2$ for all $x\in\{-1,1\}^{n}$.  Define the \textbf{influence} of the $i^{th}$ variable on $f$ to be

$$\mathrm{Inf}_{i}(f)=\langle D_{i}f,D_{i}f\rangle=2^{-n}\sum_{x\in\{-1,1\}^{n}}(D_{i}f(x))^{2}.$$
That is, $\mathrm{Inf}_{i}(f)$ is the squared $L_{2}$ norm of the $i^{th}$ partial derivative of $f$.
\end{definition}
In the case $f\colon\{-1,1\}^{n}\to\{-1,1\}$, note that $D_{i}f\in\{-1,0,1\}$.  So, if $X_{1},\ldots,X_{n}$ are i.i.d. uniform in $\{-1,1\}$, we have
$$\mathrm{Inf}_{i}(f)=\P(f(X)\neq f(X_{1},\ldots,X_{i-1},-X_{i},X_{i+1},\ldots,X_{n})),\qquad\forall\,1\leq i\leq n.$$
\begin{remark}\label{infrk}  
For any $1\leq i\leq n$, we have $\mathrm{Var}_{X_{i}}f\colonequals\E_{X_{i}}(f-\E_{X_{i}}f)^{2}$, so since $f=\sum_{S\subset\{1,\ldots,n\}}\langle f,W_{S}\rangle W_{S}$, where $W_{S}(x)\colonequals\prod_{i\in S}x_{i}$ for all $S\subset\{1,\ldots,n\}$, and $D_{i}W_{S}=0$ if $i\notin S$ and $D_{i}W_{S}=W_{S}$ if $i\in S$, we get
$$\mathrm{Inf}_{i}(f)=\E \mathrm{Var}_{X_{i}}f,\qquad\forall\,1\leq i\leq n.$$
\end{remark}

The following Theorem was conjectured in \cite{khot07} and proven in \cite{mossel10}.

\begin{theorem}[\embolden{Majority is Stablest}, {\cite{mossel10}}]\label{misthm}
Let $0\leq\rho\leq 1$ and let $0<\epsilon<1$.  Then there exists a $\tau>0$ such that the following holds, for all integers $n\geq1$.  If $f\colon\{-1,1\}^{n}\to[-1,1]$ satisfies $\E f=0$ and $\max_{1\leq i\leq n}\mathrm{Inf}_{i}(f)\leq\tau$, then
$$S_{\rho}(f)\leq\lim_{n\to\infty}S_{\rho}(\mathrm{Maj}_{n})+\epsilon=\frac{2}{\pi}\sin^{-1}(\rho)+\epsilon.$$
\end{theorem}

Theorem \ref{misthm} was deduced in \cite{mossel10} as a corollary of an invariance principle \cite[Theorem 2.1]{mossel10}, which can be understood as a nonlinear generalization of the Central Limit Theorem.
%

The invariance principle of \cite[Theorem 2.1]{mossel10} shows that the distribution of a multilinear polynomial $Q$ with small influences is almost the same when $Q$ has $\{-1-1\}$ inputs or Gaussian inputs.  A similar statement holds for the noise stability of a function itself, rather than its distribution \cite[Theorem 3.20]{mossel10}.
%

Then maximizing the noise stability of functions $f\colon\{-1,1\}^{n}\to[-1,1]$ in Theorem \ref{misthm} reduces to maximizing the noise stability of multilinear polynomials that are functions of i.i.d. Gaussian random variables.  That is, \cite[Theorem 3.20]{mossel10} implies the equivalence between an optimization over discrete functions, versus an optimization over functions on a continuous (real) domain.

Fortunately, the corresponding continuous problem was solved decades earlier, Theorem \ref{borthm} below.

\subsection{Borell's Inequality}

We define the Gaussian density as
\begin{equation}\label{zero0.0}
\begin{aligned}
\gamma_{n}(x)&\colonequals (2\pi)^{-n/2}e^{-\vnormt{x}^{2}/2},\qquad
\langle x,y\rangle\colonequals\sum_{i=1}^{n}x_{i}y_{i},\qquad
\vnormt{x}^{2}\colonequals\langle x,x\rangle,\\
&\qquad\forall\,x=(x_{1},\ldots,x_{n}),y=(y_{1},\ldots,y_{n})\in\R^{n}.
\end{aligned}
\end{equation}

Let $f\colon\R^{n}\to[0,1]$ be measurable and let $\rho\in(-1,1)$.  Define the \textbf{Ornstein-Uhlenbeck operator with correlation $\rho$} applied to $f$ by
\begin{equation}\label{oudef}
\begin{aligned}
T_{\rho}f(x)
&\colonequals\int_{\R^{n}}f(x\rho+y\sqrt{1-\rho^{2}})\gamma_{n}(y)\,\d y\\
&=(1-\rho^{2})^{-n/2}(2\pi)^{-(n)/2}\int_{\R^{n}}f(y)e^{-\frac{\vnorm{y-\rho x}^{2}}{2(1-\rho^{2})}}\,\d y,
\qquad\forall x\in\R^{n}.
\end{aligned}
\end{equation}

\begin{definition}[\embolden{Noise Stability}]\label{noisedef}
Let $\Omega\subset\R^{n}$ be measurable.  Let $\rho\in(-1,1)$.  We define the \textit{noise stability} of the set $\Omega$ with correlation $\rho$ to be
$$S_{\rho}(\Omega)\colonequals\int_{\R^{n}}1_{\Omega}(x)T_{\rho}1_{\Omega}(x)\gamma_{n}(x)\,\d x
\stackrel{\eqref{oudef}}{=}(2\pi)^{-n}(1-\rho^{2})^{-n/2}\int_{\Omega}\int_{\Omega}e^{\frac{-\|x\|^{2}-\|y\|^{2}+2\rho\langle x,y\rangle}{2(1-\rho^{2})}}\,\d x\d y.$$
Equivalently, if $X=(X_{1},\ldots,X_{n}),Y=(Y_{1},\ldots,Y_{n})\in\R^{n}$ are $n$-dimensional jointly Gaussian distributed random vectors with $\E X_{i}Y_{j}=\rho\cdot1_{(i=j)}$ for all $i,j\in\{1,\ldots,n\}$, then
$$S_{\rho}(\Omega)=\mathbb{P}((X,Y)\in \Omega\times \Omega).$$
\end{definition}

The following theorem is the continuous version of the Majority is Stablest Theorem \ref{misthm}.

\begin{theorem}[\embolden{Borell's Inequality}, \cite{borell75}]\label{borthm}
Let $\Omega\subset\R^{n}$ be measurable.  Let $H\subset\R^{n}$ be a half space (a set of points on one side of a hyperplane) such that $\gamma_{n}(\Omega)=\gamma_{n}(H)$.  Then
$$S_{\rho}(\Omega)\leq S_{\rho}(H).$$
\end{theorem}

\subsection{Sketch of the Proof of Majority is Stablest}\label{secmispf}

It is more convenient to consider functions taking values in $[0,1]$.  Fix $\mu\in(0,1)$.  To obtain an analogous statement for functions $f\colon\{-1,1\}^{n}\to[-1,1]$ with $\E f=0$, just consider $(1+f)/2$ and apply the argument below with $\mu=1/2$.

The invariance principle \cite[Theorem 3.20]{mossel10} implies that the $f\colon\{-1,1\}^{n}\to[0,1]$ with $\E f=\mu$ with largest noise stability $S_{\rho}f$ is approximately the same as taking $\Omega\subset\R^{n}$ with $\gamma_{n}(\Omega)=\mu$ with the largest noise stability $S_{\rho}(\Omega)$.  By Theorem \ref{borthm}, the largest such set is a half space $H$.  Using \cite[Theorem 3.20]{mossel10} again, $S_{\rho}(H)$ is approximately the same as $S_{\rho}g$ where $g$ is a (discrete) half space of the form $g(x)=1_{x_{1}+\cdots+x_{n}>t}$, for some $t\in\R$, for all $x=(x_{1},\ldots,x_{n})\in\{-1,1\}^{n}$.

To recover Theorem \ref{misthm}, note that $2\cdot 1_{x_{1}+\cdots+x_{n}>0}-1=\mathrm{sign}(x_{1}+\cdots+x_{n})$.  For a full proof of a more general statement, see \cite[Theorem 4.4]{mossel10}. See also Chapter 11 of \cite{odonnell14}.

\section{Plurality is Stablest Conjecture}

We now generalize our discussion from voting methods with two candidates to voting methods with three or more candidates.

Let $k\geq2$ be an integer and define the $k$-simplex to be
$$\Delta_{k}\colonequals\{(y_{1},\ldots,y_{k})\in\R^{k}\colon y_{1}+\cdots+y_{k}=1,\, y_{i}\geq0,\,\forall\,1\leq i\leq k\}.$$
For any $1\leq i\leq k$, let $e_{i}\in\R^{k}$ be the vector with a $1$ in its $i^{th}$ entry and zeros in the other entries.

A function $f\colon\{1,\ldots,k\}^{n}\to\{e_{1},\ldots,e_{k}\}$ is a \textbf{voting method} with $n$ voters and $k$ candidates.  (We associate $e_{i}$ to candidate $i$, for each $1\leq i\leq k$.)  For any $x=(x_{1},\ldots,x_{n})\in\{1,\ldots,k\}^{n}$, we think of $x_{i}$ as the vote of person $1\leq i\leq n$ for candidate $x_{i}\in\{1,\ldots,k\}$.  Given the votes $x$, the winner of the election is $f(x)$.

In this section, we denote $\langle\cdot,\cdot\rangle$ as the standard inner product on $\R^{k}$.  We denote $f_{i}\colonequals\langle f,e_{i}\rangle$.  A function $f\colon\{1,\ldots,k\}^{n}\to\Delta_{k}$ can be also be considered a (random) voting method, so that $f_{i}(x)$ is the probability that $e_{i}$ is elected, for each $1\leq i\leq k$, given votes $x$.

\begin{definition}\label{nsdef2}
For any $x\in\Z^{\sdimn}$, we denote $\vnormt{x}_{0}$ as the number of nonzero coordinates of $x$.  The \textbf{noise stability} of $g\colon\{1,\ldots,k\}^{n}\to\R$ with parameter $\rho\in(-1,1)$ is
\begin{flalign*}
S_{\rho} g
&\colonequals k^{-\sdimn}\sum_{x\in\{1,\ldots,k\}^{\sdimn}} g(x)\E_{\rho} g(Y)\\
&=k^{-\sdimn}\sum_{x\in\{1,\ldots,k\}^{\sdimn}} g(x)\sum_{y\in\{1,\ldots,k\}^{\sdimn}}\left(\frac{1-(k-1)\rho}{m}\right)^{\sdimn-\vnormt{y-x}_{0}}
\left(\frac{1-\rho}{m}\right)^{\vnormt{y-x}_{0}} g(y).
\end{flalign*}
Equivalently, conditional on $X$, $\E_{\rho}g(Y)$ is defined so that for all $1\leq i\leq\sdimn$, $Y_{i}=X_{i}$ with probability $\frac{1+(k-1)\rho}{k}$, and $Y_{i}$ is equal to any of the other $(k-1)$ elements of $\{1,\ldots,k\}$ each with probability $\frac{1-\rho}{k}$, and so that $Y_{1},\ldots,Y_{\sdimn}$ are independent.

The \textbf{noise stability} of $f\colon\{1,\ldots,k\}^{\sdimn}\to\Delta_{k}$ with parameter $\rho\in(-1,1)$ is
$$S_{\rho}f\colonequals\sum_{i=1}^{k}S_{\rho}f_{i}= \E_{\rho}\langle f(X),f(Y)\rangle.$$
\end{definition}

When $f\colon\{1,\ldots,k\}^{n}\to\{e_{1},\ldots,e_{k}\}$, we have
$$S_{\rho}f=\P_{\rho}(f(X)=f(Y)).$$

Define the \textbf{plurality} function $\mathrm{Plur}_{k,\sdimn}\colon\{1,\ldots,k\}^{\sdimn}\to\Delta_{k}$ for $k$ candidates and $\sdimn$ voters such that for all $x\in\{1,\ldots,k\}^{\sdimn}$.
$$\mathrm{Plur}_{k,\sdimn}(x)
\colonequals\begin{cases}
e_{j}&,\mbox{if }\abs{\{i\in\{1,\ldots,k\}\colon x_{i}=j\}}>\abs{\{i\in\{1,\ldots,k\}\colon x_{i}=r\}},\\
&\qquad\qquad\qquad\qquad\forall\,r\in\{1,\ldots,k\}\setminus\{j\}\\
\frac{1}{k}\sum_{i=1}^{k}e_{i}&,\mbox{otherwise}.
\end{cases}
$$

\begin{example}
When $k=2$, the Definition \ref{nsdef2} says $Y_{i}=X_{i}$ with probability $(1+\rho)/2$, and $Y_{i}\neq X_{i}$ with probability $(1-\rho)/2$, agreeing with Definition \ref{nsdisdef} (though the values that $X_{i},Y_{i}$ take are different in these two definitions).  Also, under our new definition of noise stability we have $S_{\rho}f=\P_{\rho}(f(X)=f(Y))$, so that $\lim_{n\to\infty}S_{\rho}\mathrm{Maj}_{n}=[1+(2/\pi)\sin^{-1}\rho]/2$.
\end{example}

Let $X_{1},\ldots,X_{n}$ be i.i.d. uniform in $\{1,\ldots,k\}$.  Recalling Remark \ref{infrk}, we define the \textbf{influence} of voter $1\leq i\leq n$ on $g\colon\{1,\ldots,k\}^{n}\to\R$ to be
\begin{equation}\label{infeq}
\mathrm{Inf}_{i}(g)=\E \mathrm{Var}_{X_{i}}g.
\end{equation}

\begin{conj}[\embolden{Plurality is Stablest Conjecture}]\label{pis}
For any $k\geq2$, $\rho\in[0,1]$, $\epsilon>0$, there exists $\tau>0$ such that, for all $n\geq1$, if $f\colon\{1,\ldots,k\}^{\sdimn}\to\Delta_{k}$ satisfies $\mathrm{Inf}_{i}(f_{j})\leq\tau$ for all $1\leq i\leq\sdimn$ and for all $1\leq j\leq k$, and if $\E f=\frac{1}{k}\sum_{i=1}^{k}e_{i}$, then
$$
S_{\rho}f\leq \lim_{\sdimn\to\infty}S_{\rho}\mathrm{Plur}_{k,\sdimn}+\epsilon.
$$
\end{conj}

Analogues of the invariance principles \cite[Theorem 2.1,Theorem 3.20]{mossel10} apply in this more general setting.  See \cite[Theorem 3.3, 3.4,3.6]{isaksson11}.  These invariance principles imply that Conjecture \ref{pis} is equivalent to Conjecture \ref{conj2} stated below.

\begin{prob}[\embolden{Standard Simplex Problem}, {\cite{isaksson11}}]\label{prob2}
Let $k\geq3$.  Fix $\rho\in(0,1)$.  Find measurable sets $\Omega_{1},\ldots\Omega_{k}\subset\R^{n}$ with $\cup_{i=1}^{k}\Omega_{i}=\R^{n}$ and $\gamma_{n}(\Omega_{i})=1/k$ for all $1\leq i\leq k$ that maximize
$$\sum_{i=1}^{k}\int_{\R^{n}}1_{\Omega_{i}}(x)T_{\rho}1_{\Omega_{i}}(x)\gamma_{n}(x)\,\d x,$$
subject to the above constraints.
\end{prob}

\begin{conj}[\embolden{Standard Simplex Conjecture} {\cite{isaksson11}}]\label{conj2}
Let $\Omega_{1},\ldots\Omega_{k}\subset\R^{n}$ maximize Problem \ref{prob2}.  Assume that $k-1\leq n$.  Fix $\rho\in(0,1)$.  Let $z^{(1)},\ldots,z^{(k)}\in\R^{n}$ be the vertices of a regular simplex in $\R^{n}$ centered at the origin.  Then for all $1\leq i\leq m$,
\begin{equation}\label{ome1}
\Omega_{i}=\{x\in\R^{n}\colon\langle x,z^{(i)}\rangle=\max_{1\leq j\leq k}\langle x,z^{(j)}\rangle\}.
\end{equation}
\end{conj}
%
%

At present, the only general result for Conjecture \ref{conj2} is that the optimal sets $\Omega_{1},\ldots,\Omega_{k}$ are low-dimensional, in the following sense.

\begin{theorem}[{\embolden{Dimension Reduction}, \cite{heilman20d}}]\label{pisc}
Let $\Omega_{1},\ldots\Omega_{k}\subset\R^{n}$ maximize Problem \ref{prob2}.  Assume that $k-1\leq n$.  Fix $\rho\in(0,1)$.  Then, after rotating $\Omega_{1},\ldots,\Omega_{k}$ if necessary, there exist measurable sets $\Omega_{1}',\ldots,\Omega_{k}'\subset\R^{k-1}$ such that
$$\Omega_{i}=\Omega_{i}'\times\R^{n-k+1},\qquad\forall\,1\leq i\leq k.$$
\end{theorem}

\begin{cor}[{\cite{heilman20d}}]
Let $k=3$.  Then there exists $\rho_{0}>0$ such that Conjecture \ref{conj2} is true for all $0<\rho<\rho_{0}$.
\end{cor}

\section{Ranked Choice Voting}

Let $k$ be a positive integer.  Let $\{e_{1},\ldots,e_{k}\}$ denote the standard basis of $\R^{k}$ so that $e_{i}$ has a $1$ in the $i^{th}$ entry and zeros in other entries, for all $1\leq i\leq k$.  Let $\mathcal{S}_{k}$ denote the set of permutations on $k$ elements.  That is, $\mathcal{S}_{k}$ is an ordered list of the integers $\{1,\ldots,k\}$.  For example,
$$\mathcal{S}_{3}=\{(1,2,3),(1,3,2),(2,1,3),(2,3,1),(3,1,2),(3,2,1)\}.$$

A \textbf{ranked choice voting method} with $n$ voters and $k$ candidates is a function
$$f\colon \mathcal{S}_{k}^{n}\to\{e_{1},\ldots,e_{k}\}\subset\R^{k}.$$
Let $x=(x_{1},\ldots,x_{n})\in S_{k}^{n}$.  We think of $x_{i}\in \mathcal{S}_{n}$ as the vote of person $1\leq i\leq n$.

More generally, $f\colon \mathcal{S}_{k}^{n}\to\Delta_{k}$ is a (randomized) voting method where we interpret $f(x)$ as a probability distribution on the winning candidates, given the votes $x$.

Identifying $\mathcal{S}_{k}$ with $\{1,\ldots,k!\}$ via any particular bijection, for any $1\leq i\leq n$ we can define an influence $\mathrm{Inf}_{i}(f)$ as in \eqref{infeq}.

The \textbf{noise stability} can also be defined as in the previous section, as
$$S_{\rho}f\colonequals\sum_{i=1}^{k}S_{\rho}f_{i}= \E\langle f(X),f(Y)\rangle.$$
Here $X,Y$ are defined as in Definition \ref{nsdef2}.

When $f\colon\mathcal{S}_{k}^{n}\to\{e_{1},\ldots,e_{k}\}$, we have
$$S_{\rho}f=\P(f(X)=f(Y)).$$

\begin{definition}\label{clcdef}
A permutation $\pi\in\mathcal{S}_{k}$ is a function $\pi\colon\{1,\ldots,k\}\to\{1,\ldots,k\}$.  We interpret the permutation $\pi$ as a vote, so that $\pi(1)$ is the rank of candidate $1$, $\pi(2)$ is the rank of candidate $2$, and so on.  A candidate $j\in\{1,\ldots,k\}$ is a \textbf{Condorcet loser} for a set of votes $x\in S_{k}^{n}$ if, for each $\ell\in\{1,\ldots,k\}\setminus\{j\}$,
$$\frac{1}{n}\abs{i\in\{1,\ldots,n\}\colon x_{i}(\ell)\geq x_{i}(j)}>1/2.$$
A function $f\colon\mathcal{S}_{k}^{n}\to\Delta_{k}$ satisfies the \textbf{Condorcet Loser Criterion (CLC)} if, for each $x\in S_{k}^{n}$, if $j$ is a Condorcet loser for $x$, then $f(x)\neq j$.
\end{definition}

There is an analogous condition for functions $g\colon \R^{n}\to\Delta_{k}$ with $n\geq k!$.  Suppose we have a bijection $B\colon \{1,\ldots,k!\}\to S_{k}$.  For any $1\leq i<j\leq k$, let $y^{(i,j)}\in\R^{n}$ be such that $y^{(i,j)}_{\ell}=1$ when $B(\ell)(i)>B(\ell)(j)$ and $y^{(i,j)}_{\ell}=-1$ when $B(\ell)(i)<B(\ell)(j)$, for all $1\leq\ell\leq k!$, and $y_{\ell}^{(i,j)}=0$ if $\ell>k!$.  (We also define $y^{(i,j)}\colonequals-y^{(j,i)}$ if $1\leq j<i\leq k$.) Then $g\colon \R^{n}\to\Delta_{k}$ satisfies the \textbf{Condorcet Loser Criterion (CLC)} if, for all $1\leq j\leq k$
$$\{x\in\R^{n}\colon \langle x,y^{(\ell ,j)}\rangle\geq0,\,\forall\,\ell\in\{1,\ldots,k\}\setminus\{j\}\}\cap\{x\in\R^{n}\colon g(x)=j\}=\emptyset.$$

Similarly, a partition of $\R^{n}$ into measurable subsets $\Omega_{1},\ldots,\Omega_{k}$ satisfies the \textbf{Condorcet Loser Criterion (CLC)} if, for all $1\leq j\leq k$,
$$\{x\in\R^{n}\colon \langle x,y^{(\ell ,j)}\rangle\geq0,\,\forall\,\ell\in\{1,\ldots,k\}\setminus\{j\}\}\cap\Omega_{j}=\emptyset.$$

\begin{definition}
Define the \textbf{Borda Count} function $\mathrm{Bor}_{k,\sdimn}\colon S_{k}^{\sdimn}\to\Delta_{k}$ so that, for any $x\in S_{k}^{n}$,
$$\mathrm{Bor}_{k,\sdimn}(x)
\colonequals\begin{cases}
e_{j}&,\mbox{if }\sum_{i=1}^{\sdimn}x_{i}(j)>\sum_{i=1}^{\sdimn}x_{i}(r),\,\forall\,r\in\{1,\ldots,k\}\setminus\{j\}\\
\frac{1}{k}\sum_{i=1}^{k}e_{i}&,\mbox{otherwise}.
\end{cases}
$$
\end{definition}
Observe that the Borda Count function is the plurality function applied to a vector of pairwise comparisons, i.e. if $x\in S_{k}^{\sdimn}$, we have a vector of pairwise comparisons $(\omega^{(1)},\ldots,\omega^{(\sdimn)})$ where $\omega^{\ell}\in\{1,\ldots,k\}^{\binom{k}{2}}$, if for any $1\leq\ell\leq\sdimn$, define $\omega^{(\ell)}_{\{i,j\}}\colonequals j$ if $x_{\ell}(j)>x_{\ell}(i)$ and $\omega^{(\ell)}_{\{i,j\}}\colonequals i$ if $x_{\ell}(j)<x_{\ell}(i)$.  Then $\omega\colonequals(\omega^{(1)},\ldots,\omega^{(\sdimn)})\in(\{1,\ldots,k\}^{\binom{k}{2}})^{\sdimn}=\{1,\ldots,k\}^{\sdimn\binom{k}{2}}$.  Then
$$\mathrm{Bor}_{k,\sdimn}(x)=\mathrm{Plur}_{k,n\binom{k}{2}}(\omega).$$
%
%

We then have the following adaptation of Conjecture \ref{pis} to the ranked-choice setting

\begin{conj}[\embolden{Borda Count is Stablest Conjecture}]\label{bcis}
For any $k\geq2$, $\rho\in[0,1]$, $\epsilon>0$, there exists $\tau>0$ such that, for all $n\geq1$, if $f\colon S_{k}^{\sdimn}\to\Delta_{k}$ satisfies $\mathrm{Inf}_{i}(f_{j})\leq\tau$ for all $1\leq i\leq\sdimn$ and for all $1\leq j\leq k$, if $\E f=\frac{1}{k}\sum_{i=1}^{k}e_{i}$, and if $f$ satisfies the Condorcet Loser Criterion, then
$$
S_{\rho}f\leq \lim_{\sdimn\to\infty}S_{\rho}\mathrm{Bor}_{k,\sdimn}+\epsilon.
$$
\end{conj}

As shown below in \eqref{bor1} and \eqref{bor3}, the noise stability of the Borda count function as the number of voters $n$ goes to infinity is the same as the noise stability of the Plurality function as $n\to\infty$.  So, the continuous version of Conjecture \ref{bcis} can be stated as follows.

\begin{prob}\label{sscprob}
Let $\rho\in(0,1)$.  Let $n\geq k!+k-1$.  Find a partition $\Omega_{1},\ldots,\Omega_{n}\subset\R^{n}$ be a partition satisfying $\gamma_{n}(\Omega_{i})=1/k$ for all $1\leq i\leq k$ that maximizes
$$\sum_{i=1}^{k}S_{\rho}(\Omega_{i})$$
subject to these constraints.
\end{prob}

\begin{conj}[\embolden{Standard Simplex Conjecture for Ranked Choice Voting}]\label{sscrv}
Let $\rho\in(0,1)$.  Let $n\geq k!+k-1$.  Then the sets $\Omega_{1},\ldots,\Omega_{n}\subset\R^{n}$ maximizing problem \ref{sscprob} are those defined in \eqref{ome1}.
\end{conj}

The main result of \cite{heilman20d} can be adapted to Conjecture \eqref{sscrv} as follows.

\begin{theorem}[\embolden{Dimension Reduction}]\label{bisdim}
Let $\rho\in(0,1)$.  Let $n\geq k!+k-1$.  Let $\Omega_{1},\ldots,\Omega_{k}\subset\R^{n}$ satisfy CLC and maximize Problem \ref{sscprob}.  Then, after rotating $\Omega_{1},\ldots,\Omega_{k}$, there exists $\Theta_{1},\ldots,\Theta_{n}\subset\R^{k!+k-1}$ such that
$$\Omega_{i}=\Theta_{i}\times\R^{n-[k!+k-1]},\qquad\forall\,1\leq i\leq k.$$
\end{theorem}
\begin{proof}
The argument is the same as that of \cite{heilman20d}, with the additional observation that, if $v\in\R^{n}$ is perpendicular to the vectors $y^{(i,j)}$ for all $1\leq i<j\leq n$, then the sets $\Omega_{1}+v,\ldots\Omega_{n}+v$ also satisfy CLC.
\end{proof}

\section{Plurality is Stablest Implies Borda count is Stablest}

\begin{proof}[Proof of Theorem \ref{thm0}]

Let $z^{(1)},\ldots,z^{(k)}$ be the vertices of a regular simplex centered at the origin of $\R^{k-1}$ with $\vnorm{z^{(i)}}=1$ for all $1\leq i\leq k$ and $\langle z^{(i)},z^{(j)}\rangle=-1/(k-1)$ for all $1\leq i<j\leq k$.  Recall the elementary identities $1^{2}+\cdots+k^{2}=k(k+1)(2k+1)/6$, $(1+\cdots+k)^{2}=(k(k+1)/2)^{2}$, and $(1+\cdots+k)^{2}-(1^{2}+\cdots+k^{2})=k(k+1)(k-1)(3k+2)/12$.  So, if $\pi\in S_{k}$ is uniformly random, we have
$$\E [\pi(1)]^{2}=(k+1)(2k+1)/6,\qquad \E \pi(1)\pi(2)=(k+1)(3k+2)/12.$$
\begin{equation}\label{pieq}
\E [\pi(1)]^{2}-\E \pi(1)\pi(2)=k(k+1)/12.
\end{equation}
So, using these identities and $\sum_{i=1}^{k}z^{(i)}=0$, for any $1\leq\ell\leq k$ we have
\begin{flalign*}
&\E\Big( \sum_{i=1}^{k}\pi(i)z^{(i)}\Big)\Big(\sum_{j=1}^{k}\pi(j)z^{(j)}\Big)^{T}z^{(\ell)}\\
&\qquad=\E \sum_{i=1}^{k}\pi(i)z^{(i)}\Big(\pi(\ell)-\frac{1}{k-1}\sum_{j\colon j\neq\ell}\pi(j)\Big)\\
&\qquad=z^{(\ell)}\E [\pi(1)]^{2}
+\E \pi(1)\pi(2)\sum_{i\colon i\neq\ell}z^{(i)}\\
&\qquad\qquad\qquad-\frac{1}{k-1}\sum_{j\colon j\neq\ell}\Big(\E [\pi(1)]^{2} z^{(j)}+\E \pi(1)\pi(2)\sum_{i\colon i\neq j} z^{(i)}\Big)\\
\end{flalign*}
\begin{flalign*}
&\qquad=z^{(\ell)}\Big(\E [\pi(1)]^{2}-\E \pi(1)\pi(2)\Big)-\frac{1}{k-1}\sum_{j\colon j\neq\ell}z^{(j)}\Big(\E [\pi(1)]^{2}-\E \pi(1)\pi(2)\Big)\\
&\qquad\stackrel{\eqref{pieq}}{=}\frac{k(k+1)}{12}\Big(z^{(\ell)}-\frac{1}{k-1}\sum_{j\colon j\neq\ell}z^{(j)}\Big)
=\frac{k(k+1)}{12}z^{(\ell)}\Big(1+\frac{1}{k-1}\Big)
=\frac{k^{2}(k+1)}{12(k-1)}z^{(\ell)}.
\end{flalign*}
Therefore, denoting the identity matrix as $I$,
$$
\frac{12(k-1)}{k^{2}(k+1)}
\E\Big( \sum_{i=1}^{k}\pi(i)z^{(i)}\Big)\Big(\sum_{j=1}^{k}\pi(j)z^{(j)}\Big)^{T}=I.
$$
Also, $\E  \sum_{i=1}^{k}\pi(i)z^{(i)}=0$, and if $\sigma\in S_{k}$ is such that $\sigma=\pi$ with probability $1-\epsilon$, and with probability $0<\epsilon<1$, $\sigma$ is a uniformly random element of $S_{k}$ independent of $\pi$, then
\begin{flalign*}
\E  \Big(\sum_{i=1}^{k}\pi(i)z^{(i)}\Big)\Big(\sum_{j=1}^{k}\sigma(j)z^{(j)}\Big)^{T}
&=(1-\epsilon)\E\Big( \sum_{i=1}^{k}\pi(i)z^{(i)}\Big)\Big(\sum_{j=1}^{k}\pi(j)z^{(j)}\Big)^{T}z^{(\ell)}\\
&=(1-\epsilon)\frac{k^{2}(k+1)}{12(k-1)}I.
\end{flalign*}

For any $x\in S_{k}^{n}$, define
$$V(x)\colonequals \frac{2\sqrt{3}\sqrt{k-1}}{k\sqrt{k+1}}\frac{1}{\sqrt{n}}\sum_{j=1}^{n}\sum_{i=1}^{k}x_{j}(i)z^{(i)}.$$
If $X\in S_{k}^{n}$ is uniform, and $Y\in S_{k}^{n}$ satisfies $Y_{1},\ldots,Y_{n}$ are independent, $Y_{i}=X_{i}$ with probability $1-\epsilon$ and with probability $\epsilon$, $Y_{i}\in S_{k}$ is independent of $X_{i}$ and uniformly random in $S_{k}$, then as $n\to\infty$,
$$(V(X),V(Y))$$
converges in distribution to a Gaussian random vector $(W,Z)$ where $W,Z\in\R^{k}$ are Gaussian random variables with mean, they each have identity covariance matrices, and $\E W_{i}Z_{j}=(1-\epsilon)1_{\{i=j\}}$ for all $1\leq i,j\leq k$.  Since also $\mathrm{Bor}_{k,n}(x)=j$ if and only if $\langle V(x),z^{(j)}\rangle=\max_{1\leq i\leq k}\langle V(x),z^{(i)}\rangle$ if and only if $V(x)\in\Omega_{j}$ where $\Omega_{1},\ldots,\Omega_{k}$ are defined in \eqref{ome1},
\begin{equation}\label{bor1}
\lim_{n\to\infty}S_{\left(1-\epsilon\right)}\mathrm{Bor}_{k,n}=\sum_{i=1}^{k}S_{(1-\epsilon)}(\Omega_{i}).
\end{equation}
Similarly, we have \cite[Lemma 7.3]{isaksson11}
\begin{equation}\label{bor2}
\lim_{n\to\infty}S_{\left(1-\epsilon\right)}\mathrm{Plur}_{k,n}=\sum_{i=1}^{k}S_{(1-\epsilon)}(\Omega_{i}).
\end{equation}
Together, \eqref{bor1} and \eqref{bor2} imply that
\begin{equation}\label{bor3}
\lim_{n\to\infty}S_{\left(1-\epsilon\right)}\mathrm{Bor}_{k,n}
=\lim_{n\to\infty}S_{\left(1-\epsilon\right)}\mathrm{Plur}_{k,n}.
\end{equation}
Equation \ref{bor3} concludes the proof, since if the Plurality is Stablest Conjecture holds (for $k$ candidates), the plurality function itself is an upper bound for the noise stability of a ranked choice voting method for $k$ candidates.

For completeness, we review the proof of \eqref{bor2}.  Modifying the calculation above, we have $\E z^{(\pi(1))}=0$, $\E z^{(\pi(1))}[z^{(\pi(1))}]^{T}=\frac{k}{k-1}I$, $\E z^{(\pi(1))}[z^{(\sigma(1))}]^{T}=(1-\epsilon)\frac{k}{k-1}I$, so
$$W(x)\colonequals \sqrt{\frac{k-1}{k}}\frac{1}{\sqrt{n}}\sum_{j=1}^{n}z^{(\pi(1))}$$
satisfies $(W(X),W(Y))$ also converges in distribution as $n\to\infty$ to a Gaussian random vector $(W,Z)$ where $W,Z\in\R^{k}$ are Gaussian random variables with mean, they each have identity covariance matrices, and $\E W_{i}Z_{j}=(1-\epsilon)1_{\{i=j\}}$ for all $1\leq i,j\leq k$.  Finally, $\mathrm{Plur}_{k,n}(x)=j$ if and only if $\langle W(x),z^{(j)}\rangle=\max_{1\leq i\leq k}\langle W(x),z^{(i)}\rangle$ if and only if $W(x)\in\Omega_{j}$ where $\Omega_{1},\ldots,\Omega_{k}$ are defined in \eqref{ome1}, so \eqref{bor2} holds.
\end{proof}

\section{Other Variants of the Borda count is Stablest Conjecture}

\begin{remark}
Under the assumption that voters rank candidate independently and uniformly at random, there are at least three natural restrictions to impose on ranked choice voting methods:
\begin{itemize}
\item Condorcet Loser Criterion
\item Condorcet (Winner) Criterion
\item Smith Criterion
\end{itemize}
\end{remark}
Above, we considered maximizing the noise stability of ranked choice voting methods that satisfy the Condorcet Loser Criterion, and we conjectured that Borda count is the most noise stable ``democratic'' voting method.  Below, we present an analogue if this conjecture where we instead restrict to voting methods satisfying the Condorcet (Winner) Criterion.  Borda count does not satisfy this condition, and we do not have a guess for which voting method satisfying this criterion is the most noise stable.

\subsection{Condorcet (Winner) Criterion}

A permutation $\pi\in\mathcal{S}_{k}$ is a function $\pi\colon\{1,\ldots,k\}\to\{1,\ldots,k\}$.  We interpret the permutation $\pi$ as a vote, so that $\pi(1)$ is the rank of candidate $1$, $\pi(2)$ is the rank of candidate $2$, and so on.  A candidate $j\in\{1,\ldots,k\}$ is a \textbf{Condorcet winner} for a set of votes $x\in S_{k}^{n}$ if, for each $\ell\in\{1,\ldots,k\}\setminus\{j\}$,
$$\frac{1}{n}\abs{i\in\{1,\ldots,n\}\colon x_{i}(\ell)\leq x_{i}(j)}>1/2.$$
A function $f\colon\mathcal{S}_{k}^{n}\to\Delta_{k}$ satisfies the \textbf{Condorcet Winner Criterion (CWC)} or simply the \textbf{Condorcet Criterion} if, for each $x\in S_{k}^{n}$, if $j$ is a Condorcet winner for $x$, then $f(x)= j$.

There is an analogous condition for functions $g\colon \R^{n}\to\Delta_{k}$ with $n\geq k!$.  Suppose we have a bijection $B\colon \{1,\ldots,k!\}\to S_{k}$.  For any $1\leq i<j\leq k$, let $y^{(i,j)}\in\R^{n}$ be such that $y^{(i,j)}_{\ell}=1$ when $B(\ell)(i)>B(\ell)(j)$ and $y^{(i,j)}_{\ell}=-1$ when $B(\ell)(i)<B(\ell)(j)$, for all $1\leq\ell\leq k!$, and let $y_{\ell}^{(i,j)}=0$ if $\ell>k!$.  (We also define $y^{(i,j)}\colonequals-y^{(j,i)}$ if $1\leq j<i\leq k$.) Then $g\colon \R^{n}\to\Delta_{k}$ satisfies the \textbf{Condorcet Winner Criterion (CWC)} if, for all $1\leq j\leq k$
$$\{x\in\R^{n}\colon \langle x,y^{(\ell ,j)}\rangle\leq0,\,\forall\,\ell\in\{1,\ldots,k\}\setminus\{j\}\}\subset\{x\in\R^{n}\colon g(x)=j\}.$$

Similarly, a partition of $\R^{n}$ into subsets $\Omega_{1},\ldots,\Omega_{k}$ satisfies the \textbf{Condorcet Winner Criterion (CWC)} if, for all $1\leq j\leq k$,
$$\{x\in\R^{n}\colon \langle x,y^{(\ell ,j)}\rangle\leq0,\,\forall\,\ell\in\{1,\ldots,k\}\setminus\{j\}\}\subset\Omega_{j}.$$

\begin{prob}[\embolden{Stablest Ranked Choice Voting for CWC}]\label{cwcprob}
Find a sequence of functions $g=g_{k,n}\colon S_{k}^{n}\to\Delta_{k}$ satisfying CWC such that the following holds, where $\E g=\frac{1}{k}\sum_{i=1}^{k}e_{i}$ and $\mathrm{Inf}_{i}(g_{j})\leq\tau$ for all $1\leq i\leq\sdimn$ and for all $1\leq j\leq k$.

For any $k\geq2$, $\rho\in[0,1]$, $\epsilon>0$, there exists $\tau>0$ such that, for all $n\geq1$, if $f\colon S_{k}^{\sdimn}\to\Delta_{k}$ satisfies $\mathrm{Inf}_{i}(f_{j})\leq\tau$ for all $1\leq i\leq\sdimn$ and for all $1\leq j\leq k$, if $\E f=\frac{1}{k}\sum_{i=1}^{k}e_{i}$, and if $f$ satisfies the Condorcet Winner Criterion, then
$$
S_{\rho}f\leq \lim_{\sdimn\to\infty}S_{\rho}g_{k,\sdimn}+\epsilon.
$$
\end{prob}

\medskip
\noindent\textbf{Acknowledgement}.  Thanks to Jake Freeman, Elchanan Mossel and Alex Tarter for helpful discussions.

\bibliographystyle{amsalpha}

\newcommand{\etalchar}[1]{$^{#1}$}
\def\polhk#1{\setbox0=\hbox{#1}{\ooalign{\hidewidth
  \lower1.5ex\hbox{`}\hidewidth\crcr\unhbox0}}} \def\cprime{$'$}
  \def\cprime{$'$}
\providecommand{\bysame}{\leavevmode\hbox to3em{\hrulefill}\thinspace}
\providecommand{\MR}{\relax\ifhmode\unskip\space\fi MR }
\providecommand{\MRhref}[2]{%
  \href{http://www.ams.org/mathscinet-getitem?mr=#1}{#2}
}
\providecommand{\href}[2]{#2}

\end{document}